\documentclass[english,a4paper,12pt]{article}
\usepackage[utf8]{inputenc}
\usepackage[T1]{fontenc}
\usepackage{titling}
\usepackage{amsmath, amssymb}
\usepackage{pgf}
\usepackage{amsfonts}
\usepackage{dsfont}
\usepackage{mathrsfs}
\usepackage{mathabx}
\usepackage{stmaryrd}
\usepackage{makeidx}
\usepackage{amsbsy}
\usepackage{amsthm}
\usepackage{color}
\usepackage{fullpage}
\usepackage{enumitem}
\usepackage[vcentermath]{youngtab}
\usepackage{marginnote} %%% margin note
\usepackage{mdframed}	%%% implement corrections
\newmdenv[topline=false, bottomline=false, skipabove=\topsep, skipbelow=\topsep]{siderules}

\usepackage{tikz}
\usetikzlibrary{decorations.text,calc,arrows.meta}
\usepackage[vcentermath]{youngtab}
\usepackage{tkz-euclide}
\usetkzobj{all}
\usepackage{pgfplots}
\usetikzlibrary{shapes}
\usetikzlibrary{decorations.shapes}
\usetikzlibrary{positioning}

\def\ben{\begin{equation}}
\def\een{\end{equation}}
\def\bena{\begin{eqnarray}}
\def\eena{\end{eqnarray}}

\def\d{{\rm d}}
\def\V{{\cal V}}
\def\S{{\mathcal S}}

\def\S{\mathbf{S}}

\def\cP{{\mathcal P}}

\def\path{\operatorname{Pa}}

\def\b{\operatorname{b}}

\def\H{\mathcal{H}}

\newcommand{\gM}{{\mathfrak A}}
\newcommand{\gB}{{\mathfrak B}}

\newcommand{\myid}{{\bf 1}}

\renewcommand{\SS}{{\mathbb S}}

\newcommand{\RR}{{\mathbb R}}

\newcommand{\CC}{{\mathbb C}}
\newcommand{\ZZ}{{\mathbb Z}}

\def\D{{\mathcal D}}

\renewcommand{\b}{{\mathbf b}}

\newtheorem{thm}{Theorem}
\newtheorem{remark}{Remark}

\newtheorem{lemma}{Lemma}

\newtheorem{corollary}{Corollary}%[section]
\newtheorem{defn}{Definition}%[section]

\begin{document}
%DOCUMENT PROPER
\title{Relative entropy for coherent states in chiral CFT}
\author{Stefan Hollands$^1$\thanks{\tt stefan.hollands@uni-leipzig.de} \\ \\
{\it $^{1}$ Institute for Theoretical Physics, University of Leipzig,}\\
{\it Br\"{u}derstra{\ss}e 16, D-04103 Leipzig, Germany.} 
	}
\date{15 May 2019}

\maketitle 

\begin{abstract}
We consider the relative entropy between the vacuum state and a state 
obtained by applying an exponentiated stress tensor to the vacuum of a chiral conformal field theory on the lightray.  
The smearing function of the stress tensor can be viewed as a vector field on the real line generating a diffeomorphism. 
We show that the relative entropy is equal to $c$ times the so-called Schwarzian action of the diffeomorphism. As an application 
of this result, we obtain a formula for the relative entropy between the vacuum and a solitonic state. 
\end{abstract}

\section{Introduction}

Quantum information theoretic considerations in quantum field theory have attracted a lot of attention in recent years, not least due to intriguing relations 
with quantum field theory in curved spacetime or even (quantum) gravity theory, for instance through the ``quantum focussing conjecture'', its relation with Bekenstein bounds \cite{Longo1}, the ``quantum null energy condition'' \cite{bousso1,bousso2}, ``c-theorems'' \cite{Casini2} and many other topics. See e.g. the book \cite{Hubeny} for an exposition of holographic ideas in this context.

In this note we study chiral conformal field theories (CFTs), i.e. a 1-dimesional chiral half of a full $1+1$-dimensional CFT living on one lightray. We consider the the vacuum state $|\Omega \rangle$, and a state $|\Omega' \rangle = \exp(i \int \Theta(u) f(u) \d u) |\Omega\rangle$, where $f(u)$ is a real valued smooth testfunction such that $f(0)=0$ and  
$\Theta(u)$ is the stress tensor on the light-ray. 
We think of $|\Omega' \rangle$ as the analogue of a coherent state in the CFT. 

Let $\gM$ be the subalgebra of all observables consisting of those which are 
localized on the positive real half line. Then we obtain partial states corresponding to $|\Omega \rangle$ resp. 
$|\Omega' \rangle$ with respect to this sub-algebra. 
We formally denote their reduced density matrices by ``$\rho = {\rm Tr}_{\RR_-} |\Omega \rangle \langle \Omega |$'' and ``$\rho'
= {\rm Tr}_{\RR_-} |\Omega' \rangle \langle \Omega' |$'' and we consider their relative entropy
$S(\rho | \rho') = {\rm Tr}[\rho(\log \rho - \log \rho')]$. It is an information theoretic measure of the indistinguishability of the 
two states with respect to observers occupying the positive half line.
Our aim is to demonstrate that this quantity -- defined rigorously with operator algebraic methods -- is equal to the Schwarzian action associated with the function $f(u)$. More precisely, consider the map $s \mapsto e^s = u$ from the real line to the positive reals. Under this map, $f(u)$, viewed as a vector field, transforms 
to $F(s)=e^{-s} f(e^s)$. Let $\varphi_t(s)$ be the flow of this vector field, i.e. the solution to $\d \varphi_t(s)/ \d t = F(\varphi_t(s))$ and $\varphi_0(s)=s$.
We write $\varphi_1(s)=\varphi(s)$. Then we shall prove that
\ben
\label{eq:main}
S(\rho | \rho ') =  \frac{c}{24}  \int_\RR \left(   \varphi'(s)^2 
+  \left( \frac{\varphi''(s)}{\varphi'(s)} \right)^2 
- 1  \right) \d s = c \, I_{\rm Schwarz}(\varphi), 
\een
where $c$ is the central charge. 
$I_{\rm Schwarz}$ is the so-called ``Schwarzian action'', which has appeared in a number of other contexts, for instance the SYK-model, holographic description of the Jackiw-Teitelboim dilaton gravity, or open string theory, see e.g. \cite{1,2,3,4,5,6,7,11,12}. 

The relative entropy between the vacuum state and coherent states (of a different nature than those considered here) 
has also recently appeared in works by Casini et al. \cite{Casini} and Longo \cite{Longo} in the context of a free scalar field. While these works are more 
restrictive in that only a free theory is considered, they are more general in that the results hold also for a non-zero mass and arbitrary dimensions. Some interesting formulae for relative entropies in conformal field theory for states generated by a local primary have also been obtained by \cite{Lashkari}, pointing perhaps to a generalization of our result to other states\footnote{For localized states in the case of a $U(1)$-current field, see e.g. \cite{Longo2}. Replica method computations of relative entropies of states excited by some local operator in various CFTs have been given carried out by \cite{Calabrese1,Calabrese2}.}. In view of a wider 
potential significance for gravity, it would be desirable to clarify the connection between our result to ideas from holography such as in \cite{8,9,10}.

It has been observed long ago \cite{KW} that chiral conformal field theories arise naturally when studying quantum field theories on 
spacetimes with a bifurcate Killing-horizon, as the restriction, in some sense, of the field theory to a lightlike generator of the horizon.
In this context, $u$ corresponds to the affine parameter of the generator, $s=\kappa t$, 
where $t$ is the ``Killing parameter'' related to the horizon Killing field, and $\kappa$ is the surface gravity associated with the 
Killing horizon. Due to 
the difference between $t$ and $s$, there would now appear a further prefactor $1/\kappa$ in front of $I_{\rm Schwarz}$.

This paper is organized as follows. In sec. \ref{sec:CFT} we recall our setting for CFT and some known results used in the sequel. 
In sec. \ref{sect1}, we recall the definition of the relative entropy. Then we prove our main result \eqref{eq:main} at an increasing level of generality in the remaining sections. Our final theorem is thm. \ref{thm:2}. This theorem allows one to apply this formula to 
solitonic states in CFT (cor. \ref{cor:1}).

\section{Notation and CFT basics}
\label{sec:CFT}

In this review section we describe our notation and basic facts about the stress tensor in two-dimensional conformal field theories (CFTs). 
The material is standard, and more details may be found e.g. in~\cite{Fe&Ho05}. We employ the operator formalism for CFT. In this formalism, the stress energy operator in a CFT on $(1+1)$-dimensional Minkowski spacetime has two independent, commuting (``left and right chiral'') components. These depend only on the left and right moving light-ray coordinates $u=x^0-x^1, v=x^0+x^1$, respectively. In this paper, we focus on only one of them. It is  
a quantum field $\Theta(u)$ living on one of the light-rays. As is well-known, a light-ray may be compactified to a circle via the stereographic projection (Cayley transform), 
and in this way we get a 
quantum field on the circle. We distinguish it notationally by $T(z)$. In some sense it is actually most natural 
to turn this story around and start from the quantum field $T(z)$ on the circle, which we shall do now. 

The starting point is the Virasoro algebra, i.e. the Lie-algebra with generators  $\{ L_n, C \}_{n \in \ZZ}$  obeying
\ben
[L_n,L_m] = (n-m) L_{n+m} + \frac{1}{12} n(n^2-1) \delta_{n,-m} C, \quad 
[L_n, C] = 0 .
\een  
A positive energy representation on a Hilbert space $\H$ is a representation such that 
\begin{enumerate}
\item[(i)]
$L^*_n=L_{-n}$ (unitarity), 
\item[(ii)] $L_0$ is diagonalizable with 
non-negative eigenvalues of finite multiplicity, and 
\item[(iii)] the central element is represented by $C = c \myid$. 
\end{enumerate}
From now, we assume such a positive energy representation. We assume furthermore that $\H$ 
contains a vacuum vector $|\Omega_0\rangle$ which is annihilated by $L_{-1}, L_0, L_1$, 
($\mathfrak{sl}(2,\RR)$-invariance) and which is a highest weight vector (of weight 0), i.e. $L_n |\Omega_0\rangle = 0$ for all 
$n >0$. In \cite{carpi,buchholz,goodman_wallach,goodman_wallach1} one can find proofs for the bound
\ben
\label{poly}
\|(1+L_0)^k L_n \Psi \| \le \sqrt{c/2}(|n|+1)^{k+3/2} \|(1+L_0)^{k+1} \Psi \|
\een
for $|\Psi \rangle \in \V$, with $\V:=\bigcap_{k \ge 0} \D (L_0^k) \subset \H$ and any natural number $k$.

From the Virasoro algebra one defines the stress tensor on the unit circle $\SS$, identified with points $z=e^{i\theta}, \theta \in \RR$ in $\CC$ by 
\ben
T(z) = -\frac{1}{2\pi} \sum_{n =-\infty}^\infty L_n z^{-n-2} . 
\een
The series must be understood in the sense of distributions. This means that the mathematically meaningful object is the smeared field
which, for a smooth testfunction $f \in C^\infty(\SS)$ is defined by 
\ben
T(f) = \int_{\SS} T(z) f(z) \d z :=-\frac{1}{2\pi} \sum_{n =-\infty}^\infty  \left(\int_{\SS} z^{-n-2} f(z) \d z\right)L_n. 
\een
\eqref{poly} thereby gives the 
inequality $\| L_0^k T(f) \Psi \| \le (c/2)^{1/2} \| f \|_{W^{3/2+k,1}(\SS)} \, \|(1+L_0)^{k+1}\Psi\|$, where $W^{s,1}(\SS)$ 
is the Sobolev space\footnote{The norm of $W^{s,p}$ is $\|f\|_{W^{s,p}}= 
(\sum_{n \in {\mathbb Z}} |\hat f_n|^p (1+|n|)^{ps})^{1/p}$, where $\hat f_n$ is the $n$-th
Fourier coefficient.} of order $(s,1)$. The estimate shows that  
$L_0^kT(f)|\Psi \rangle$ is a well defined vector of $\H$ for any $f \in W^{3/2,1}(\SS)$ for any vector $|\Psi\rangle$ in the
subspace $\D(L_0^{k+1})$. It follows that $\V$ is a dense invariant domain for $T(f)$ if $f \in C^\infty(\SS)$. 

If we define $\Gamma$ to be the anti-linear involution 
\begin{equation}
\Gamma f(z) = -z^2 \overline{f(z)},
\end{equation} 
then the smeared stress tensor is a symmetric, and in fact a (closable) essentially self-adjoint- 
operator on any core of the operator $L_0$ (such as $\D(L_0)$ or $\V$), for $f \in W^{3/2,1}(\SS)$ 
obeying the reality condition $\Gamma f = f$, see thm. 4.4 of \cite{CW}. In particular, we may 
use the functional calculus to define the unitary operators $e^{iT(f)}$ and 
$\V$ is an invariant domain for any of these operators for $f \in C^\infty$.

To make the connection with the representations of the diffeomorphism group on the circle, we first note that any 
real test function (in the above sense) $f$ defines a real vector field ${\sf f}
\in {\rm Vect}_\RR(\SS)$ by means of the formula 
\ben\label{ff}
({\sf f} g)(z) = f(z) g'(z), 
\een
where $ie^{i\theta}g'(e^{i\theta}) = \frac{\d}{\d \theta} g(e^{i\theta})$. In particular, if we 
define $l_n(z) = z^{n+1}$ then the corresponding complex vector fields ${\sf l}_n = z^{n+1} \frac{\d}{\d z} \in {\rm Vect}_\CC(\SS) = {\rm Vect}_\RR(\SS) \otimes_\RR \CC$ satisfy the Witt algebra
\ben
[{\sf l}_n, {\sf l}_m]=(m-n) {\sf l}_{n+m}
\een
under the commutator of vector fields, and furthermore $iT(l_n) = L_n$.  

For real $f \in C^\infty(\SS)$, we denote by ${\rm Exp}(tf)=\rho_t \in {\rm Diff}_+(\SS)$ the 1-parameter flow of orientation preserving 
diffeomorphisms generated by the corresponding vector field ${\sf f}$. In other words, $\rho_t$ is uniquely determined by the conditions
\ben\label{flow}
\frac{\partial}{\partial t} \rho_t(z) = f(\rho_t(z)), \quad \rho_0 = \rm{id}, 
\een 
and $\rho_t$ leaves invariant all $z$ outside the support of $f$. The unitary operators $e^{iT(f)}$ can be thought of as representers of the 
the diffeomorphisms ${\rm Exp}({\sf f})$. More precisely, as shown in Prop. 5.1 by \cite{Fe&Ho05}, there exists a strongly continuous 
unitary representation $\widetilde{{\rm Diff}_+(\SS)} \owns \rho \mapsto V(\rho) \in {\mathfrak B}(\H)$ which {\bf V1)} leaves invariant the vectors 
from the dense domain $\V$, {\bf V2)}  satisfies the composition law 
\ben
\label{eq:comp}
V(\rho_1)V(\rho_2) = \exp(icB(\rho_1, \rho_2))V(\rho_1 \circ \rho_2), 
\een
and {\bf V3)} 
$\frac{\d }{\d t} V({\rm Exp}(t{\sf f})) |\Psi \rangle_{t=0} = iT(f)|\Psi \rangle$ on a dense domain of vectors $|\Psi\rangle$ 
(such as, e.g., $\V$), and this domain coincides exactly with the domain for which (the closure of) $T(f)$ is self-adjoint. Here
\ben
\label{eq:Bott}
B(\rho_1, \rho_2) = -\frac{1}{48\pi} {\rm Re} \int_{\SS} \log (\rho_1 \circ \rho_2)'(z) \frac{\d}{\d z} \log \rho_2'(z) \, \d z
\een
is the Bott ``2-cocycle''. [An $n$-{\bf cocycle} on a group $G$ is a mapping $C_n$ from $G^n \to {\mathcal X}$ where $\mathcal X$ carries a 
representation $D$ of $G$, such that $\b C_n=0$, where 
\ben
\begin{split}
\b C_n(g_1, \dots, g_{n+1}) =& C_n(g_1, \dots, g_n) + \sum_{m=1}^n (-1)^m C_n(g_1, \dots g_mg_{m+1}, \dots g_n) \\
&+ (-1)^{n+1} D(g_1) C_n(g_2, \dots, g_{n+1}) 
\end{split}
\een
is called the coboundary operation.]
It follows that $V({\rm Exp}(t{\sf f}))=e^{i\alpha(t)} e^{it\Theta(f)}$ for a phase $\alpha(t)$ (depending on $f$) such that $\alpha'(0)=0$. 
See prop. 5.1 of \cite{Fe&Ho05} for the non-trivial global assignment of the phases, which corrects an error in \cite{goodman_wallach}.

For a smooth function $f(z)$ on 
the complex plane or circle, the Schwarzian derivative is defined by 
\ben
\label{eq:S}
\S f(z) = \frac{f'''(z)}{f'(z)} - \frac{3}{2} \left( \frac{f''(z)}{f'(z)} \right)^2 . 
\een
Using the representation theoretic facts stated above, it has been shown in~\cite{Fe&Ho05} (Prop. 3.1 and Sec. 5.2), which uses results of \cite{goodman_wallach,goodman_wallach1,toledano},
that on the domain $\V$, we have the transformation formula
\ben\label{rel1}
V(\rho) T(z) V(\rho)^* = \rho'(z)^2 T(\rho(z)) - \frac{c}{24 \pi} \S\rho (z) \, \myid , 
\een
to be understood in the sense of distributions in the variable $z \in \SS$. Here 
where $\rho = \rho_{t=1}$ is the flow of $f$ at unit flow-`time', i.e., $\rho= {\rm Exp}({\sf f})$, and in the above formula, we 
could also replace $V(\rho)$ by $e^{iT(f)}$, since both operators only differ by a phase.

Going back to the real line, the stress tensor on $\RR$ is defined by pulling back the stress tensor on the circle via the Cayley transform $C$. It maps
the real line to the circle by 
\begin{equation}
\RR \owns u \mapsto C(u)=\frac{1+iu}{1-iu} \in \SS \setminus \{-1\},
\end{equation} 
with inverse $C^{-1}(z)=i(1-z)/(1+z)$. Defining $C(\infty)=-1$ this becomes a bijection between $\RR \cup \infty$ and the circle.
Then, by definition, the stress tensor on the real line is
\ben\label{real}
\Theta(u) \equiv  \left( \frac{\d C(u)}{\d u} \right)^2 T(C(u)) = -\frac{4}{(1-iu)^4} T\left(\frac{1+iu}{1-iu} \right).  
\een 
It is an operator valued distribution on the same domain. Using this formula, we can easily convert any result on stress tensor on 
the circle to one on the real line. 

\section{Relative entropy}\label{sect1}

We first recall the definition of the relative entropy in terms of modular operators. For details and background on 
operator algebras, see \cite{Bratteli}. For a recent survey of operator algebraic methods in quantum information theory in
QFT, we refer to \cite{sanders_2}. An exposition directed towards theoretical physicists is \cite{Witten:2018zxz}.

A v. Neumann algebra, $\gM$, is an algebra of bounded operators on a Hilbert space $\H$ that is 
closed in the topology induced by the size of matrix elements.\footnote{We always assume that $\H$ is separable.} 
We assume that $\H$ contains a ``cyclic and separating'' vector for $\gM$, meaning a unit vector $|\Omega \rangle$
such that the set consisting of $A|\Omega\rangle$, $A \in \gM$ is a dense subspace of $\H$, and such that $A|\Omega\rangle=0$ always implies $A=0$. We say in this case that $\gM$ is in ``standard form'' with respect to the given vector. 

The ``Tomita operator'', $S$, is the anti-linear operator on the domain $\D(S) = \{ A|\Omega\rangle \mid A \in \gM\}$ defined by 
\ben
SA|\Omega \rangle = A^* |\Omega \rangle . 
\een
$S$ is defined consistently due to the cyclic and separating property. $S$ is a closable operator\footnote{We denote its closure by the same symbol.} and as such, $S$ has a polar decomposition denoted by $S=J\Delta^\frac12$, with $J$ anti-linear and unitary and $\Delta$ self-adjoint and non-negative. 
Tomita-Takesaki theory is about the properties of the operators $\Delta, J$ and related objects, the basic ones of which are: (i) $J$ exchanges 
$\gM$ with the commutant $\gM'$ (the set of all bounded operators on $\H$ 
commuting with all operators in $\gM$) in the sense $J \gM J = \gM'$. Furthermore $J^2 = 1, J\Delta J = \Delta^{-1}$, 
(ii) The ``modular flow'' $\sigma_t(A) =  \Delta^{it} A \Delta^{-it}$ leaves $\gM$ invariant in the sense $\sigma_t \gM = \gM$ for all $t \in \RR$,
(iii) The positive, normalized (meaning $\omega(A^*A) \ge 0 \, \, \forall A \in \gM, \omega(1) = 1$) 
linear {\bf expectation functional} 
\ben 
\omega(A) = \langle \Omega |A\Omega \rangle
\een
is invariant under $\sigma^t$ and $j={\rm Ad}(J)$ and satisfies the so-called KMS condition. (iv)
Any normal (i.e. continuous in the weak$^*$-topology) positive linear functional $\omega'$ on $\gM$ has a unique vector representative 
$|\Omega' \rangle$ in the ``natural cone'' 
\ben
\cP^\sharp 
=\overline{\{ Aj(A) |\Omega \rangle \mid A \in \gM\}}, 
\een
where the overbar means closure and $j(A) = JAJ$. The state functional is thus
$\omega'(A) = \langle \Omega' |A\Omega' \rangle$ for all $a \in \gM$. 

One can generalize this construction to that of the relative modular operator, flow etc. Let $\omega'$ be another normal state on $\gM$, and let $|\Omega' \rangle$ be its unique vector representative in the natural cone in $\H$, which is assumed (for simplicity) to be cyclic and separating, too. 
Then it is consistent to define
\ben
S_{\omega, \omega'} A |\Omega' \rangle = A^* |\Omega \rangle
\een
and make the polar decomposition $S_{\omega, \omega'}=J_\omega^{} \Delta^\frac12_{\omega,\omega'}$. 
One shows that the ``relative modular operator'' $\Delta_{\omega,\omega'}$ is related to $\Delta_\omega$ by 
$\Delta^{it}_{\omega',\omega}=u(t) \Delta^{it}_{\omega}$ where $u(t)$ is a unitary from $\gM$ called the ``Connes cocycle''. 
It is specified uniquely by the cocycle relation $u(t+s)=u(t) \sigma_\omega^t(u(s))$ (which fixes $u(t)$ up to a phase $e^{it\alpha}$)
and the KMS condition ${\rm anal.cont.}_{t\to i}\omega(u(t))=1$ (which fixes the phase). The Connes cocycle is customarily denoted by 
$u(t)=(D\omega:D\omega')_{it}$ since this notation correctly suggests a chain rule. 
The {\bf Araki relative entropy} is defined by\footnote{The derivative on the right side 
is understood in the usual way as the limit 
$i \lim_{t \to 0} t^{-1}\omega(u(t)-1)$.
By thm. 5.7 of \cite{Petz}, the limit exists if $S(\omega | \omega')<\infty$.} 
\ben
S(\omega | \omega') = i \frac{\d}{\d t} \omega(u(t)) \bigg|_{t=0} = \langle \Omega | (\log  \Delta_{\omega,\omega'}) \Omega \rangle . 
\een
In the case of finite Type I factors (``quantum mechanics''), e.g. $\gM=M_N(\CC)$, the situation is this: State functionals are equivalent to 
density matrices $\rho_\omega$ via $\omega(A) = {\rm Tr}(A\rho_\omega)$, $\H$ is  the algebra itself $M_N(\CC) \cong \CC^N \otimes \CC^N$ on which $\gM$ acts by left multiplication. 
The state $|\Omega\rangle$ corresponds to $|\rho^{1/2}_\omega \rangle$, the inner product is $\langle A | B \rangle = {\rm Tr}(A^* B)$, the modular operator is
$\Delta_\omega = \rho_\omega \otimes \rho^{-1}_\omega$, and the relative modular operator is $\Delta_{\omega,\omega'}=\rho_\omega \otimes \rho^{ -1}_{\omega'}$. 
The Connes cocycle is $u(t)=\rho_{\omega'}^{it} \rho_{\omega}^{-it}$. Using this, one immediately verifies
$S(\omega | \omega') = {\rm Tr} \rho_\omega ( \log \rho_\omega - \log \rho_{\omega'})$. These formulae do not hold for 
type III factors which occur in quantum field theories, but are very useful to gain intuition.

The relative entropy and its properties are discussed in detail  e.g. in \cite{Petz}. It is never negative, but can be infinite, is decreasing under completely positive maps, is jointly convex in both arguments, to name a few. The physical interpretation of $S(\omega | \omega')$ is: $\exp[-NS(\omega | \omega')]$ is asymptotically equal to the probability of mistaking the true state of the system, $\omega$,  to be $\omega'$ having performed $N$ measurements of observables from $\gM$.

In this paper, we will need in particular the special case when
\ben
\omega'(A) \equiv \omega(U^* A U) = \langle U \Omega | A U\Omega \rangle,  
\een
where $U$ is some unitary operator from $\gM$. 
The corresponding vector representative in the natural cone is $|\Omega_U \rangle=Uj_\omega (U)|\Omega \rangle$, with $j_\omega (A)= J_\omega A J_\omega$.
It is easily seen to be cyclic and separating. 
Going through the definitions, one finds immediately that $j_\omega(U) \Delta^{1/2}_\omega j_\omega(U^*) = \Delta^{1/2}_{\omega, \omega'}$, and 
one finds that 
$
u(t)=U \Delta^{it} U^* \Delta^{-it}
$
implying that
\ben\label{Srel}
S(\omega | \omega') = i \frac{\d}{\d t} 
\langle \Omega | U \Delta^{it} U^* \Delta^{-it} | \Omega \rangle
 \bigg|_{t=0}  = -\langle U^* \Omega |  (\log \Delta) U^*\Omega \rangle, 
\een
where $\Delta$ is the modular operator of the original state $\omega$.

\section{Relative entropy between vacuum and coherent state}

We now apply \eqref{Srel} to the following situation: $\gM$ is the v. Neumann algebra $\gM_{\rm Vir}(\RR_+)= \{e^{i\Theta(g)} | g \in C^\infty_{0,\RR}(\RR_+)\}''$ associated with the smeared stress tensor on the right half-line $\RR_+$, $|\Omega\rangle = |\Omega_0\rangle$ is the vacuum state, and $U=e^{i\Theta(f)}$, where $f$ is a specific smooth compactly supported real valued function. 
Our aim is to prove \eqref{eq:main}. We will proceed by considering situations of in creasing generality in subsections \ref{sec:c1} 
and \ref{sec:c2}. At first, we assume that $f$ is supported $(k,\infty)$, where $k>0$. In this case, the argument is elementary. Subsequently, we will weaken the 
assumptions on the behavior of $f$ at $u=0$. 

To begin, we note that $|\Omega_0\rangle$ is separating (and cyclic) by the Reeh-Schlieder theorem, so the results of Tomita-Takesaki theory apply. 
We call $\omega(A)=\langle \Omega_0|A\Omega_0 \rangle$ resp. $\omega'(A)=\langle \Omega_0|U^* A U\Omega_0 \rangle$ 
the corresponding expectation functionals on $\gM_{\rm Vir}(\RR_+)$.
The modular operator $\Delta$ for $|\Omega_0\rangle$ with respect to $\gM_{\rm Vir}(\RR_+)$ is known from the Bisognano-Wichmann theorem \cite{bisognano} to be given by $\Delta^{it} A \Delta^{-it} = e^{-2\pi i t D}Ae^{2\pi i D}$, where $D=-\frac{i}{2}(L_1-L_{-1}) = \int u \Theta(u) \d u$ is the generator of dilations on the real line and $A \in \gM_{\rm Vir}(\RR_+)$. Then \eqref{Srel} gives, since $U = e^{i\Theta(f)} \in \gM_{\rm Vir}(\RR_+)$ and since the vacuum is invariant under $\Delta^{it}$ and under $e^{2\pi i t D}$:
\ben\label{Srel1}
\begin{split}
S(\omega |  \omega') &=i \frac{\d}{\d t} \langle \Omega_0 | U \Delta^{it} U^* \Delta^{-it} \Omega_0 \rangle \bigg|_{t=0} \\
&=i \frac{\d}{\d t} \langle \Omega_0 | U  e^{-2\pi it D} U^* e^{2\pi it D} \Omega_0 \rangle \bigg|_{t=0} \\
&=i \frac{\d}{\d t} \langle \Omega_0 | e^{i\Theta(f)} e^{-2\pi it D} e^{-i\Theta(f)} \Omega_0 \rangle \bigg|_{t=0} 
=2\pi \, \langle \Omega_0 |  e^{i\Theta(f)} D e^{-i\Theta(f)}  \Omega_0 \rangle ,
\end{split}
\een
using in the last step that $e^{-i\Theta(f)} |\Omega_0 \rangle \in \V=\bigcap_{k \ge 0} \D (L_0^k)$, which is  
a common core for the operators $\Theta(g)$, and in particular $D$. This follows from Theorem VIII.11 of \cite{Reed} and footnote 8 of \cite{Fe&Ho05}. 

At this stage, we are allowed to use the transformation formula pulled back to the real line via the Caley transform \eqref{rel1}  and the remark below that equation to obtain 
\ben\label{rel2}
e^{i\Theta(f)} D e^{-i\Theta(f)} = V(\rho) D V(\rho)^* =  \int_{-\infty}^\infty u \rho'(u)^2 \Theta(\rho(u)) \, \d u - \frac{c}{24 \pi}  \int_{-\infty}^\infty u \S\rho (u) \,  \d u , 
\een
where $\rho(u)={\rm Exp}(f)(u)$ is the diffeomorphism of the positive half-line generated by the vector field $f(u) \d / \d u$, i.e. $\rho(u) = \rho_1(u)$ 
for the flow  $\d \rho_t(u) / \d t=f(\rho_t(u))$ and $\rho_0(u)=u$. 
This diffeomorphism is by construction equal to the identity on the left 
real half-line, meaning that the second term on the right side of this equation requires integration only over the positive real half line. The first term involving 
the stress tensor has a vanishing expectation value in the vacuum, therefore 
\ben
\label{eq:Su}
S(\omega |  \omega')=  -\frac{c}{12}  \int_{0}^\infty u \S\rho (u) \,  \d u . 
\een
The  term on the right side can be rewritten in a more suggestive form introducing 
$u=e^s$ and $\varphi(s) = \log(\rho(e^s))$. Geometrically, $\varphi$ is a diffeomorphism on the real line induced by $\rho$ under the exponential map.
Therefore, it can be viewed as the flow $\varphi_t(s)$ at unit time $t=1$ of the pull-back of the vector field $f(u)$ given by $F(s) = e^{-s} f(e^s)$, i.e. 
$\d \varphi_t(s) / \d t=F(\varphi_t(s))$.
Using repeatedly the chain rule 
\ben
\S(H \circ G) = (G')^2 (\S H) \circ G + \S G
\een
for the Schwarzian derivative of $\rho= \exp \circ \varphi \circ \log$, we have the elementary calculation
\ben
\begin{split}
\S\rho(u) 
&= (\varphi \circ \log)'(u)^2 (\S\exp) \circ \varphi(\log u) + \S(\varphi \circ \log)(u)\\
&= -\tfrac{1}{2} u^{-2} [\varphi' \circ \log (u)]^2 + \S(\varphi \circ \log)(u)\\
&=  -\tfrac{1}{2} u^{-2}[\varphi' \circ \log(u)]^2 + u^{-2} (\S\varphi) \circ \log (u) + (\S \log)(u) \\
&=u^{-2} \{-\tfrac{1}{2} [\varphi'(\log u)]^2 +  (\S\varphi) (\log u) + \tfrac{1}{2}\},
\end{split}
\een
and therefore, since $u=e^s$, 
\ben
u^2 \S\rho(u) 
= \S\varphi(s) + \tfrac{1}{2} -  \tfrac{1}{2} [\varphi'(s)]^2
\een
Thus, using $\d u/u=\d s$, we obtain from \eqref{eq:Su}, upon substitution of the Schwarzian derivative \eqref{eq:S}
\ben
\label{eq:thm1}
\begin{split}
S(\omega |  \omega') &= - \frac{c}{12}  \int_{-\infty}^\infty \left( \S\varphi(s) + \tfrac{1}{2} -  \tfrac{1}{2} [\varphi'(s)]^2 \right) \d s \\
&=\frac{c}{12}  \int_\RR \left(   \frac{1}{2} \varphi'(s)^2 
-\left( \frac{\varphi''(s)}{\varphi'(s)} \right)' 
+ \frac{1}{2} \left( \frac{\varphi''(s)}{\varphi'(s)} \right)^2 
- \frac{1}{2}  \right) \d s .  
\end{split}
\een
This gives the claim made in \eqref{eq:main} since we can drop the total derivative term under the last integral, as $f$ has compact support on the positive 
half-axis $\RR_+$, so $\varphi''(s)/\varphi'(s)=0$ outside a compact set. Another way to write the relative entropy is 
\ben
\label{eq:S0}
S(\omega |  \omega') =  \frac{c}{24} \int_0^\infty u \left( \frac{\d}{\d u} \log \rho'(u) \right)^2 \d u \ , \quad \rho = {\rm Exp}(f), 
\een
and this makes manifest the general property $S(\omega | \omega') \ge 0$. We summarize the result in a theorem.

\begin{thm}
\label{thm:1}
Let $\omega$ be the vacuum, 
$f \in C^\infty_{0,\RR}(\RR)$, with $f(u)=0$, when $u<\epsilon$ for some $\epsilon>0$. Let 
$\omega'(A)=\omega(e^{-i\Theta(f)} A e^{i\Theta(f)})$ be the corresponding coherent state, $\rho = {\rm Exp}(f)$ the corresponding diffeomorphism (so $\rho(0)=0, \rho'(0)=1$), and $\varphi(s) = \log \rho(e^s)$. Then the relative entropy with respect to  $\gM_{\rm Vir}(0,\infty)$ is
\ben
S(\omega | \omega') = c\,   I_{\rm Schwarz}(\varphi) .
\een
\end{thm}

The theorem admits an obvious generalization to conformal nets which are finite index extensions of the net $\{\gM_{\rm Vir}(I)\}$, see 
\cite{KawahigashiLongo} for a complete classification in the case $c<1$ and for the precise definitions. For us, it is only important that such a net $\{\gB(I)\}$ is labelled by open intervals $I \subset \SS$, and that for each such interval $\gB(I)$ is a v. Neumann factor
containing $\gM_{\rm Vir}(I)$ such that the Jones index $[\gB(I):\gM_{\rm Vir}(I)]$ is finite. By the Pimsner-Popa theorem, this entails the 
existence of a ``conditional expectation'' $E_I: \gB(I) \to \gM_{\rm Vir}(I)$, that is a linear map such that $E_I(B^*B) \ge k B^*B$ for $k>0$  for any $B \in \gB(I)$, and such that $E_I(A_1BA_2)=A_1E_I(B)A_2$ for $A_1,A_2 \in \gM_{\rm Vir}(I)$. There is a unique ``minimal'' $E_I$ such that $k$ is largest (actually equal to the inverse squared Jones index).

We can pull the nets back to the real line via the Caley transform. The minimal conditional expectation $E: \gB(\RR_+) \to \gM_{\rm Vir}(\RR_+)$ defines a faithful normal state on $\gB(\RR_+)$ by $\psi:=\omega \circ E$, where $\omega$ is the vacuum on  
$\gM_{\rm Vir}(\RR_+)$. Then we have:

\begin{corollary}
Let $\psi'$ be the coherent state defined  from $\psi$ as in thm. \ref{thm:1}. Then $S_{\gB(\RR_+)}(\psi | \psi') = c \, I_{\rm Schwarz}(\varphi)$. 
\end{corollary}

\begin{proof}
We have $S_{\gB(\RR_+)}(\psi | \psi') = S_{\gB(\RR_+)}(\omega\circ E | \omega' \circ E) = S_{\gM_{\rm Vir}(\RR_+)}(\omega | \omega' )$, where the last equality is a standard property of the relative entropy for conditional expectations $E$, see e.g. thm. 5.19 of \cite{Petz}. The result thereby follows from thm. \ref{thm:1}.
\end{proof}

\subsection{Generalization: Case 1}
\label{sec:c1}
It is natural to ask what happens if $f \in C^\infty_{0,\RR}(\RR)$ is of compact support but the support is not necessarily on the positive half-axis $\RR_+$. In 
this case, $U=e^{i\Theta(f)}$ does not necessarily belong  to the algebra $\gM_{\rm Vir}(0,\infty)$ of the positive half-axis, and formula \eqref{Srel} may not be valid. It is easy to see that, regardless of the choice of $f$, as long as the support of $f$ is compact, the relative entropy between the vacuum and the corresponding coherent 
state is always finite. To this end, consider a $k>0$ such that the support of $f$ is in $(-k,\infty)$. Now, the relative entropy $S_{(-k,\infty)}(\omega | \omega')$ relative to the algebra $\gM_{\rm Vir}(-k,\infty) \supset \gM_{\rm Vir}(0,\infty)$ is not less than 
$S_{(0,\infty)}(\omega | \omega')$ (by monotonicity, see e.g. thm. 5.3 of \cite{Petz}), 
and can be computed by means of thm. \ref{thm:1}, because the vacuum is translation invariant 
and we can thereby effectively shift $k$ to $0$. Thus, in particular $S_{(-k,\infty)}(\omega | \omega')<\infty$ and therefore also 
$S_{(0,\infty)}(\omega | \omega')<\infty$.

However, this does not tell us how to actually compute the relative entropy. But 
we can adapt the preceding argument at least in certain cases. 
First, we restrict attention in this subsection to the special case when $f(0)=f'(0)=0$, which in view of 
the differential equation \eqref{flow}
defining $\rho = {\rm Exp}(f)$ 
is the same as saying that $\rho(0)=0, \rho'(0)=1$. Then we can write $\rho=\rho_- \circ \rho_+=\rho_+ \circ \rho_-$, where 
\ben
\label{eq:rpm}
\rho_\pm(u) = 
\begin{cases}
\rho(u) & \text{if $u \in \RR_\pm$,}\\
u & \text{if $u \in \RR_\mp \cup \{0\}$.}
\end{cases}
\een
The diffeomorphisms $\rho_\pm$, while smooth away from $u=0$, are only 
of differentiablity class $C^1$ at $u=0$. Now, if we could nevertheless 
show that the unitary representer $V(\rho)$ of $\rho$ factors as $V(\rho_-)V(\rho_+)$ up to a phase, 
then we could expect that $V(\rho_\pm)$ are elements of the algebras $\gM_{\rm Vir}(\RR_\pm)$ of the right/left half axis. 
If true, the Connes cocycle would be 
\ben
\label{eq:Connes0}
(D\omega : D\omega')_{it} = V(\rho_+) \Delta^{it} V(\rho_+)^* \Delta^{-it} . 
\een
 Then, if we could show that $V$ still forms a representation up to phases of the group 
 ${\rm Diff}_+^{{\rm ps},1}(\SS)$ of $C^1$ and piecewise 
smooth diffeomorphisms on $\SS$ (respectively on $\RR$ via the Caley transform\footnote{We will write simply $\rho \in {\rm Diff}_+^{{\rm ps},1}$ if we want to leave it open whether $\rho$ should be viewed as a function on $\RR$ or $\SS$. In the former case it is 
understood that $\rho$ should be $C^1$ at the point $u=\infty$.}) with composition law \eqref{eq:comp}, we could write 
$V(\rho_+) \Delta^{it} V(\rho_+)^* \Delta^{-it}=V(\rho_+ \circ \delta_t \circ \rho_+^{-1} \circ \delta_t^{-1}) e^{icB(\rho_+ , \delta_t \circ \rho_+^{-1} \circ \delta_t^{-1})}$
with $\delta_t(u)=e^{-2\pi t}u$ denoting a dilation of the lightray. So it would follow
\ben
\label{eq:Connes1}
(D\omega : D\omega')_{it} = V(\rho_+ \circ \delta_t \circ \rho_+^{-1} \circ \delta_t^{-1}) e^{icB(\rho_+ , \delta_t \circ \rho_+^{-1} \circ \delta_t^{-1})} . 
\een 
Finally if we are allowed to differentiate this formula inside an expectation value with respect to $|\Omega_0\rangle$,   
 then the above argument would go through unchanged, 
and we would again obtain \eqref{eq:thm1}. The total derivative term still does not give a boundary term at $s=-\infty$ ($u=0$), 
and so we would still obtain \eqref{thm:1} under the assumption that $f(0)=f'(0)$. 

However it is not totally obvious that these arguments really go through because the diffeomorphisms $\rho_\pm$, while smooth away from $u=0$, are only of differentiablity class $C^1$ at $u=0$. The composition formula \eqref{eq:comp} needs to be justified for non-smooth diffeomorphims. Such questions have been dealt with in great detail by \cite{CW,delVeccio}, and all the relevant arguments are in fact available in the literature. We will therefore be brief. 

If $g$ is a piecewise smooth, real, compactly supported $C^1$-function on $\SS$, then by standard arguments, $g \in W^{s,1}$
for any $s<2$. Therefore, by prop. 4.4 of \cite{CW}, $T(g)$ is a closable essentially self-adjoint operator (e.g. on $\D(L_0)$). 
We can hence define $e^{iT(g)}$ in the usual way by the spectral theorem. Furthermore, if $g_n \to g$ in $W^{3/2,1}$, 
$e^{iT(g_n)} \to e^{iT(g)}$ in the strong operator topology (prop. 4.5 of \cite{CW}). 

These facts already suffice to make the above argument rigorous. First, let $f^\pm(u)=0$ for $u\in \RR_\mp$ and $f^\pm(u)=f(u)$ 
for $u \in \RR_\pm$. Let $f_n^\pm$ be a sequence of real valued $C_0^\infty$ functions such that $f_n^\pm \to f^\pm$ in $W^{3/2,1}$
and such that the support of $f^\pm_n$ is strictly contained in $\RR_\pm$. Then $e^{i\Theta(f_n^+)+i\Theta(f^-_n)} = e^{i\Theta(f_n^+)}
e^{i\Theta(f^-_n)}$. Both sides strongly converges to $e^{i\Theta(f)} = e^{i\Theta(f^+)}
e^{i\Theta(f^-)}$, and $e^{i\Theta(f^\pm)} \in \gM_{{\rm Vir}}(\RR_\pm)$.  This already justifies \eqref{eq:Connes0}. 
It is also straightforward to justify the analogue of \eqref{Srel}. First we evaluate 
\ben
\begin{split}
&i \frac{\d}{\d t} \omega(e^{i\Theta(f_n^+)} e^{-2\pi itD} e^{-i\Theta(f_n^+)} e^{2\pi itD}) \\
= &2\pi 
\omega\bigg(e^{i\Theta(f_n^+)} e^{-2\pi itD} e^{-i\Theta(f_n^+)} [\Theta(g_n) -  c\beta({\rm Exp}(f^+_n), u) {\bf 1}]\bigg),
\end{split}
\een
where $g_n = {\rm Exp}(f^+_n)_* u$, and where $\beta(\rho, f)=1/(24\pi )\int_\RR f(u) \S \rho (u) \d u$.
We know $e^{\pm i\Theta(f_n^+)} \to e^{\pm i\Theta(f^+)}$ converge strongly and $g_n$ converges in $C^1$ to a piecewise smooth function, hence in $W^{3/2,1}$. Thus, $\Theta(g_n) |\Omega_0\rangle$ converges in $\H$.  
Likewise, $\beta({\rm Exp}(f^+_n), u)$
also converges using the explicit formula for the Schwarzian derivative $\S$.
Therefore, the right side converges uniformly in $t$, implying that the derivative 
at $t=0$ can be evaluated before taking the limit $n \to \infty$. We conclude in view of \eqref{Srel}, \eqref{eq:Connes0}
that \eqref{eq:Su} holds in the present setting. Hence, thm. \ref{thm:1} still holds for $f \in C^\infty_0$ whose support is 
may include the negative real axis but still $f(0)=f'(0)=0$. 

\medskip

While this accomplishes our goal, we want to show in the remainder of this section that, for certain piecewise smooth $C^1$ diffeomorphisms $\gamma$ of $\SS$, we can define unitary representers $V(\gamma)$ up to phase in such a way  that 
V1)--V3) of sec. \ref{sec:CFT}, and in particular \eqref{eq:comp}, 
holds. This will establish eq. \eqref{eq:Connes1} for the Connes cocycle, and will also be of use in the next section. 
To fix notation, we make the following definitions:

\begin{defn}
\label{def:1}
 ${\rm B}_0$ is the group of diffeomorphisms $\gamma$ of $\RR \cup \infty = \SS$ which are smooth except at $u=0$ and $C^1$ everywhere, and such that $\gamma(0)=0, \gamma'(0)=1$. The Lie-algebra ${\frak b}_0$ is the space 
of $C^1$ functions $g(u)$ which are smooth except $u=0$ and such that $g(0)=g'(0)=0$. 
\end{defn}

We can equip ${\rm B}_0$ with a natural 
topology giving it the structure of a Frechet manifold as follows. If $\gamma \in {\rm B}_0$ 
then $\log \gamma(e^{i\theta}) - \theta$ is a $C^1$ map which is smooth except at $\theta=0$ where it has derivative 0, and period $2\pi$. Vice versa, any such $\gamma$ can be obtained in this way. We see that ${\rm B}_0$ gets identified with a 
convex subset of $C^{1,{\rm ps}}_{2\pi}(\RR, \RR)$. This space can be equipped with a family of semi-norms providing it with the structure of a Frechet space in a natural way. 

Next, we recall that it has been shown in lemma A.3 of \cite{delVeccio} that any $\gamma \in {\rm B}_0$ can be decomposed as ${\rm Exp}(g) \circ \underline{\gamma}$, with $g \in {\mathfrak b}_0$ and with $\underline{\gamma} \in {\rm Diff}_+^\infty$. The construction 
of $g$ depends on certain arbitrary choices, but by going through the proof of lemma A.3 of \cite{delVeccio}, one 
can see that locally near a given $\gamma$, the assignment of $g$ (hence of $\underline{\gamma}$) is $C^1$ in 
the topology of ${\rm B}_0$ (and the induced topology on ${\frak b}_0$. 
One can then define a unitary transformation implementing $\gamma$ by $U(\gamma):=e^{iT(g)} V(\underline{\gamma})$, and \cite{delVeccio} have shown that this gives a projective representation of ${\rm B}_0$ (prop. A.4). Furthermore, if $|\Psi\rangle \in \V$, 
the map ${\rm B}_0 \owns \gamma \mapsto U(\gamma)|\Psi\rangle \in \H$ is differentiable. This can be seen by applying the Duhamel 
formula ($f \in T_\gamma {\rm B}_0, D$ the directional derivative):
\ben
\begin{split}
D_f U(\gamma)|\Psi\rangle =& \int_0^1 e^{i(1-s)T(g)} T(D_f g) e^{isT(g)} V(\underline \gamma) |\Psi\rangle \, \d s +
e^{iT(g)} D_f V(\underline \gamma) |\Psi\rangle \\
=& e^{iT(g)} \int_0^1 \bigg( T({\rm Exp}(-sg)_* D_f g) + c\beta({\rm Exp}(-sg), D_f g) {\bf 1} \bigg) V(\underline \gamma) |\Psi\rangle \, \d s +\\
& e^{iT(g)} D_f V(\underline \gamma) |\Psi\rangle . 
\end{split}
\een
We note that all terms on the right side are well defined since ${\rm Exp}(sg)_* D_f g(\gamma) \in {\frak b}_0 \subset W^{3/2,1}$
with continuous dependence on $s$, and since $\gamma \mapsto V(\underline \gamma) |\Psi\rangle$ is differentiable by the 
results of \cite{Fe&Ho05}, sec. 5. To go to the 
second line, we have used the transformation formula \eqref{rel1}, justified in the present case by lemma 3.3 of \cite{delVeccio}
(which thereby in particular shows that the subspace $\{ U(\gamma) |\Psi\rangle \, | \, \Psi \in \V, \gamma \in {\rm B}_0\}$
is in the domain of any $T(f)$, when $f \in {\frak b}_0$).   

Next, we want to change the representers $U(\gamma)$ to new representers $V(\gamma)$ by assigning additional phases
in such a way that \eqref{eq:comp} holds. It  is plausible that this ought to be possible because
 the Bott-cocycle $B(\gamma_1,\gamma_2)$ is still well-defined when $\gamma_1, \gamma_2 \in {\rm B}_0$. We can see it e.g. from 
the explicit form \eqref{eq:Bott}, because $(\gamma_2 \circ \gamma_2)', \gamma_2'$ are smooth up to a finite number of points where they are continuous, and therefore the integrand is locally like the product of a step function times a continuous function. 
That such an assignment of phases is possible was shown in 
prop. 5.1 of \cite{Fe&Ho05} for the case of smooth diffeomorphisms. 
The main difference is that in the present case, the assignment $\gamma \mapsto U(\gamma)$ (providing the local charts of 
the bundle $\widehat {\rm B}_0$ in the language of \cite{Fe&Ho05}), has been shown only to be once differentiable, rather than smooth. However, inspecting the argument given in prop. 5.1 of \cite{Fe&Ho05} shows that it is essentially of topological nature, so 
differentiablity is enough to reach the same conclusion (with phases on ${\rm B}_0$ that are once differentiable). 
Furthermore, as also shown in 
prop. 5.1 of \cite{Fe&Ho05}, we may achieve that $\frac{\d }{\d t} V({\rm Exp}(tf)) |\Psi \rangle_{t=0} = iT(f)|\Psi \rangle$ 
e.g. on the dense domain $\{ V(\gamma) |\Psi\rangle \, | \, \Psi \in \V, \gamma \in {\rm B}_0\}$. To summarize, we have:

\begin{lemma}
\label{lem:0}
We have V1)-V3) (see sec. \ref{sec:CFT})  for ${\rm B}_0 \owns \gamma \mapsto V(\gamma)|\Psi\rangle \in \H$, which is differentiable,
on the dense domain $\{ V(\gamma) |\Psi\rangle \, | \, \Psi \in \V, \gamma \in {\rm B}_0\}$.
\end{lemma}

\begin{remark}
One can generalize this reasoning without difficulty to the group of $C^1$ diffeomorphisms $\gamma$ of $\SS$ (or $\RR \cup \infty$ via the Caley transform) which are smooth 
except for a finite set $\{z_1, \dots, z_N\}$ of given points where $\gamma(z_i)=z_i, \gamma'(z_i)=1$. We will apply this below to 
the case when the points are $0, \infty$, and we will still call the group ${\rm B}_0$ by abuse of notation. 
\end{remark}

\subsection{Generalization: Case 2}
\label{sec:c2}
Now we relax our assumption and consider smooth $f$ on $\RR$. At $u=\infty$ we demand that $f$ and its derivative vanish. 
At $u=0$ we demand $f(0)=0$ but not necessarily $f'(0)=0$. Then $\rho={\rm Exp}(f)$ leaves $u=\infty$ invariant and is $C^1$ 
there\footnote{Here as well as everywhere, the Caley transform is understood to identify $\RR \cup \infty$ and $\SS$.}.
At $u=0$, we have $\rho(0)=0$ but not necessarily $\rho'(0)=1$, and so $\rho_\pm$ as in \eqref{eq:rpm}
are only continuous but not differentiable at $0$. As a consequence, the 
split $V(\rho_+\rho_-)=e^{icB(\rho_+,\rho_-)} V(\rho_+)V(\rho_-)$ (FALSE) cannot be defined in this case 
since even the Bott cocycle is undefinined (infinite) for diffeomorphisms of the class ${\rm Diff}_+^{{\rm ps},0}$
This can easily be seen from eq. \eqref{eq:Bott}, because in the case at hand both $\rho_2'(u), (\rho_1 \circ \rho_2)'(u)$ have a finite discontinuity at $u=0$, thus the integrand of \eqref{eq:Bott} behaves qualitatively like the product of a $\delta$-function and 
a step function. Such a product is not well-defined, although it can be extended, non-uniquely, as a distribution. In particular, therefore,
the phase $B(\rho_+ , \delta_t \circ \rho_+^{-1} \circ \delta_t^{-1})$ appearing in \eqref{eq:Connes1} is undefined. 

Nevertheless we can try 
to define the Connes cocycle $(D\omega : D\omega')_{it}$ by a modification of \eqref{eq:Connes1}, making the ansatz
\ben
\label{eq:Connes2}
(D\omega : D\omega')_{it} = 
V(\rho_+ \circ \delta_t \circ \rho_+^{-1} \circ \delta_t^{-1}) e^{ia(t)} =: u(t) ,  
\een 
 where $a(t) \in \RR$ is to be determined. Note that the right side is well defined, since even though  $\rho_+$ is only in ${\rm Diff}_+^{{\rm ps},0}$, the 
 combination $\rho_+ \circ \delta_t \circ \rho_+^{-1} \circ \delta_t^{-1}$ is in ${\rm B}_0 \subset {\rm Diff}_+^{{\rm ps},1}$ (cf. defn. \ref{def:1}), so 
 $V(\rho_+ \circ \delta_t \circ \rho_+^{-1} \circ \delta_t^{-1})$ is 
 well defined by lemma \ref{lem:0} and the following remark, 
 which follows easily from the fact that $\delta_t(u)=e^{-2\pi t}u$ are dilations in the light ray picture. The importance of $a(t)$ is that, in view of \eqref{Srel} and lemma \ref{lem:0}, the relative entropy is given by
 \ben
 \label{eq:Salpha'}
S(\omega | \omega') = -a'(0). 
\een
 The derivative, defined in the usual way as the limit $a'(0)=\lim_{t\to 0} t^{-1}[a(t)-a(0)]$, must exist in view of thm. 5.7 of \cite{Petz}, because we have already argued that $S(\omega | \omega')$ is finite.
 
 Now we would like to impose the Connes cocycle condition 
 $u(t+s)=u(t) \sigma_\omega^{t} (u(t))$. This is equivalent to 
 \ben
 \label{eq:cocy}
 a(t)+a(s) - a(t+s) = b(t,s)
 \een
 where 
 \ben
 \label{eq:beta}
 b(t,s) = c B\left([\rho_+, \delta_t], \delta_t \circ [\rho_+, \delta_s] \circ \delta_t^{-1} \right) , 
 \een
 using the usual notation $[g_1,g_2]=g_1 g_2 g^{-1}_1 g_2^{-1}$ for the commutator in a group (here the diffeomorphism group).
 This is well defined since $[\rho_+, \delta_t] \in {\rm B}_0$.
 We can rewrite this condition as $\b a= b$, where $\b$ is the cocycle operator on the additive group $(\RR,+)$.
 Since there are non non-trivial cocycles on this group, solutions $a$ can be found provided that $\b b = 0$. 
 We state this as a 
 \begin{lemma}
 Viewed as a 2-cycle on the additive group $(\RR,+)$, $b$ verifies $\b b = 0$.
 \end{lemma}
 \begin{proof}
 We compute
 (omitting the factor $c$ from the second line onwards):
 \ben
 \begin{split}
 &\b b(t,s,r) \\
 &= b(t,s) - b(s+t,r)+ b(t,r+s)-b(s,r)\\
 &=B([\rho_+, \delta_t],\delta_t \circ [\rho_+, \delta_s] \circ \delta_t^{-1}) 
  - B([\rho_+, \delta_{s+t}],\delta_{s+t} \circ [\rho_+, \delta_r] \circ \delta_{s+t}^{-1})\\
 &+B([\rho_+, \delta_t],\delta_t \circ [\rho_+, \delta_{s}] \circ \delta_t^{-1}) 
 - B([\rho_+, \delta_s],\delta_s \circ [\rho_+, \delta_r] \circ \delta_s^{-1}) \\
 &=B([\rho_+, \delta_t],\delta_t \circ [\rho_+, \delta_s] \circ \delta_t^{-1}) 
- B([\rho_+, \delta_{t}] \circ \delta_t \circ [\rho_+, \delta_{s}] \circ \delta_t^{-1},\delta_{s+t} \circ [\rho_+, \delta_r] \circ \delta_{s+t}^{-1})\\
 &+B([\rho_+, \delta_t],\delta_t \circ [\rho_+, \delta_{s}] \circ [\rho_+, \delta_{r}] \circ \delta_{t+s}^{-1}) 
 - B([\rho_+, \delta_s],\delta_s \circ [\rho_+, \delta_r] \circ \delta_s^{-1}) \\
 &=B([\rho_+, \delta_t],\delta_t \circ [\rho_+, \delta_s] \circ \delta_t^{-1}) 
- B([\rho_+, \delta_{t}] \circ \delta_t \circ [\rho_+, \delta_{s}] \circ \delta_t^{-1},\delta_{s+t} \circ [\rho_+, \delta_r] \circ \delta_{s+t}^{-1})\\
 &+B([\rho_+, \delta_t],\delta_t \circ [\rho_+, \delta_{s}] \circ \delta_t^{-1} \circ \delta_{s+t} \circ [\rho_+, \delta_{r}] \circ \delta_{s+t}^{-1}) 
 - B(\delta_t\circ [\rho_+, \delta_s] \circ \delta_t^{-1},\delta_{s+t} \circ [\rho_+, \delta_r] \circ \delta_{s+t}^{-1}) \\
 &=0
 \end{split}
 \een 
 using $\b B(g_1,g_2,g_3)=0$ in the last step.
 \end{proof}
 
 Every solution to \eqref{eq:cocy}, which we now know exists, is unique up to a 1-cocycle on $(\RR,+)$, i.e. up to $a(t) \to a(t)+m t$ for some $m \in \RR$.
 
 The cocycle $u(t)$ in eq. \eqref{eq:Connes2}, and $a(t)$, depend on $\rho$ and we emphasize this by writing $u_\rho(t), a_\rho(t)$ etc. 
 Let $\sigma \in {\rm B}_0$, acting nontrivially only inside $(0,1)$ (so $\sigma(0)=0, \sigma'(0)=1$).
 We now evaluate $u_{\sigma \circ \rho}(t)$ in two different ways. The first way is to use the definition \eqref{eq:Connes2} and the composition law (phases)
 for the representation $V$ of ${\rm B}_0$, cf. lemma \ref{lem:0}. We get:
 \ben
 \begin{split}
 &u_{\sigma \circ \rho}(t)\\ 
 &= V([\sigma_+ \circ \rho_+, \delta_t]) e^{ia_{\sigma \circ \rho}(t)}\\
 &=V(\sigma_+ \circ [\rho_+, \delta_t] \circ \delta_t \circ \sigma_+^{-1} \circ \delta_t^{-1}) e^{ia_{\sigma \circ \rho}(t)}\\
 &=V(\sigma_+) V([\rho_+, \delta_t] ) V(\delta_t \circ \sigma_+^{-1} \circ \delta_t^{-1}) e^{ia_{\sigma \circ \rho}(t)+ic(
 B(\sigma_+,[\rho_+, \delta_t] \circ \delta_t \circ \sigma_+^{-1} \circ \delta_t^{-1} ) + B([\rho_+, \delta_t] ,\delta_t \circ \sigma_+^{-1} \circ \delta_t^{-1} )
 )}\\
 &=V(\sigma_+) V([\rho_+, \delta_t] ) \Delta^{it}V(\sigma_+)^* \Delta^{-it} e^{ia_{\sigma \circ \rho}(t)+ic(
 B(\sigma_+,[\rho_+, \delta_t] \circ \delta_t \circ \sigma_+^{-1} \circ \delta_t^{-1} ) + B([\rho_+, \delta_t] ,\delta_t \circ \sigma_+^{-1} \circ \delta_t^{-1} )
 )}.
 \end{split}
 \een
 We note that all terms are well defined since $\sigma_+, [\rho_+, \delta_t] \in {\rm B}_0$. The second way uses the general 
 formula for the Connes cocycle: $(D\omega:D\omega'\circ {\rm Ad}(U^*))_{it}=U(D\omega:D\omega')_{it} \sigma^t_\omega(U^*)$ for $U \in \gM$, 
 which one can prove e.g. using the chain rule. Now, we have by construction $\omega_{\sigma \rho}=\omega \circ  {\rm Ad}(V(\sigma \circ \rho)^*)=
 \omega \circ  {\rm Ad}(V(\rho)^*) \circ {\rm Ad}(V(\sigma)^*) = \omega_\rho \circ {\rm Ad}(V(\sigma_+)^*)$ and this gives
 \ben
  u_{\sigma \circ \rho}(t) =V(\sigma_+) V([\rho_+, \delta_t] ) \Delta^{it}V(\sigma_+)^* \Delta^{-it} e^{ia_\rho(t)}. 
 \een
We have shown:

\begin{lemma}
We have
 \ben
 \label{eq:alphacon}
 a_{\sigma \circ \rho}(t) +c(
 B(\sigma_+,[\rho_+, \delta_t] \circ \delta_t \circ \sigma_+^{-1} \circ \delta_t^{-1} ) + B([\rho_+, \delta_t] ,\delta_t \circ \sigma_+^{-1} \circ \delta_t^{-1} )) = 
 a_\rho(t)
 \een
for all $\sigma \in {\rm B}_0$, such that $\sigma(0)=0, \sigma'(0)=1$ in particular. 
\end{lemma}

To summarize what we know at this point: the Connes coycle is given by eq. \eqref{eq:Connes1}, 
with $a(t)$ restricted by eqs. \eqref{eq:cocy}, \eqref{eq:beta} and \eqref{eq:alphacon}.

We now define an $\alpha(t)$ solving \eqref{eq:cocy} with right side \eqref{eq:beta}. Define first ($r=\rho'(0)$)
\ben
\rho_n(u) = 
\begin{cases}
u & \text{if $u < 0 $,}\\
(n \log r)^{-1}(e^{n(\log r) u}-1) & \text{if $0 \le u \le 1/n$.}\\
\rho(u-1/n) - (r-1)/(n \log r) & \text{if $u > 1/n$.}
\end{cases}
\een 
Then $\rho_n \in {\rm Diff}_+^{1,{\rm ps}}$ and $\rho_n \to \rho_+$ pointwise as $n \to \infty$ by construction. Hence, comparing \eqref{eq:Connes2} with \eqref{eq:Connes1}, it is tempting to define 
\ben
\label{eq:alphadef}
a_\rho(t) := \lim_{n \to \infty} cB(\rho_n , \delta_t \circ \rho_n^{-1} \circ \delta_t^{-1}). 
\een
Since $\rho_n \in {\rm Diff}_+^{1,{\rm ps}}$, each term $a_{\rho_n}(t)$ under the limit is well-defined and using the explicit formula for the Bott-cocycle, we can see that also the limit exits. We now claim that the cocycle relation \eqref{eq:cocy} is satisfied. It is clear since 
\eqref{eq:Connes1} is a Connes-cocycle that the coboundary operator $\delta$ applied to each term $a_{\rho_n}(t)$
under the limit must yield $b_{\rho_n}(s,t)$ given by \eqref{eq:beta}. But the latter is easily seen to have limit $b_{\rho}(s,t)$ (using again the explicit form of the Bott cocycle), demonstrating that $a_\rho(t)$ is one solution to \eqref{eq:cocy}.

Of course, $a_\rho(t)$ as defined by \eqref{eq:alphadef} is only one among the infinitely many solutions to \eqref{eq:cocy} all differing 
by $m_\rho t$ for some $m_\rho \in \RR$, and we do not know a priori which $m_\rho t$, if any, we must add to $a_\rho(t)$ 
in order for \eqref{eq:Connes2} to represent the Connes cocycle. We could at this stage in principle 
fix this ambiguity by the KMS-condition for $u(t)$ in \eqref{eq:Connes2}, i.e. 
${\rm anal.cont.}_{t\to i} \omega(u(t))=1$. But while this must be possible, it does not seem very practical, because the analytic continuation, which we know must exist, is difficult to calculate. 
We therefore resort to other considerations which we shall outline now.

The key observation is that $a_\rho(t)$ 
must satisfy by itself the constraint \eqref{eq:alphacon}. This can be seen by approximating 
$\rho_+$ by $\rho_n$. For the approximating sequence, \eqref{eq:alphacon} must hold true by a direct computation using the 
cocycle condition for $B$, and the limit exists for both sides. Thus, we have learned that the ambiguity $m_\rho t$ 
satisfies
\ben
m_{\sigma \circ \rho} = m_\rho
\een
for all $\sigma \in {\rm B}_0$ -- such that $\sigma(0)=0, \sigma'(0)=1$ in particular -- and therefore $m_\rho$, which must be 
added as $a_\rho(t) \to a_\rho(t)+m_\rho t$  in order for \eqref{eq:Connes2} to represent the Connes cocycle, can be 
a function of $r=\rho'(0)$ only.

Now, taking a derivative with respect to $t$ at $t=0$, one can see by a computation that 
\ben
\begin{split}
- a_\rho'(0) &= -m_\rho + \frac{c}{24}  \lim_{n \to \infty}  \int_{0^+}^\infty u {\bf S}\rho_n(u) \d u \\
&= -m_\rho + \frac{c}{24} (\log r)^2 + \frac{c}{24}  \int_{0^+}^\infty u \left( \frac{\d}{\d u} \log \rho'(u) \right)^2 \d u .
\end{split}
\een
Since $S(\omega | \omega') = -a'(0)$ and since $m_\rho$ is a function of $r=\rho'(0)$ only, 
we know at this stage that 
\ben
\label{eq:Snu}
S(\omega | \omega') = \nu(r)  + \frac{c}{24}  \int_{0^+}^\infty u \left( \frac{\d}{\d u} \log \rho'(u) \right)^2 \d u
\een
for some function $\nu(r)$ that we must still determine.

The major advantage of expression \eqref{eq:Snu} is that we can test it with any $\rho$ satisfying the assumptions of this section, 
with $\rho(0)=0$ and a given $r=\rho'(0^+)$,
against known general properties of the relative entropy. For instance, we know $S(\omega | \omega') \ge 0$. To see what we 
can learn from this, consider the sequence
\ben
\sigma_n(u) = 
\begin{cases}
\gamma_n(u) & \text{if $u < 0 $,}\\
\frac{\log n}{ \log(n/r)} [(u+1/n)^{\log(n/r)/\log n} - (1/n)^{\log(n/r)/\log n} ] & \text{if $0 \le u \le 1- 1/n$,}\\
u + k_n & \text{if $u > 1-1/n$,}
\end{cases}
\een
where the constant $k_n$ is chosen so that $\sigma_n$ is $C^1$ at $u=1-1/n$, and where 
$\gamma_n$ is a smooth diffeomorphism chosen such that $\sigma_n$ is smooth at $u=0$ and such that
$\gamma_n(u)=u$ for $u<-1$. $\sigma_n$ is then compactly supported, $C^1$ including at $u=\infty$, and piecewise smooth. 
At $u=0$ we have $\sigma_n(0) = 0, \sigma_n'(0)=r$. Thus we can form the corresponding 
coherent states $\omega'_n(A)=\omega(V(\sigma_n)^* A V(\sigma_n))$. If we mollify the non-smooth behavior at $u=1-1/n$, 
then we are allowed to apply \eqref{eq:Snu} to compute the entropy $S(\omega | \omega'_n) \ge 0$. The choice of 
$\gamma_n$ is immaterial for the result as it affects only the negative real axis, and the mollifier can be taken 
away without problem using that the integral in \eqref{eq:Snu1} converges absolutely and the dominated convergence theorem. 
Thereby we obtain
\ben
\begin{split}
\label{eq:Snu1}
0 &\le \nu(r)  + \frac{c}{24}  \int_{0^+}^\infty u \left( \frac{\d}{\d u} \log \sigma_n'(u) \right)^2 \d u\\
&=\nu(r)  +  \frac{c(\log r)^2}{24(\log n)^2}  \int_{0}^{1-1/n} \frac{u \d u }{(u+1/n)^2}  \\
&<\nu(r) +  \frac{c(\log r)^2}{24 \log n} . 
\end{split}
\een
But $n$ can be made arbitrarily large. So we conclude our first inequality: $\nu(r) \ge 0$. 

We next obtain a bound from above. We consider 
\ben
\zeta_n(u) = 
\begin{cases}
u & \text{if $u < -1/n $,}\\
-1/n+\int_{-1/n}^u \exp[(\log r)(nx+1)^{1/n}] \d x & \text{if $-1/n \le u < 0$,}\\
\sigma_n(u) + m_n & \text{if $u \ge 0$,}
\end{cases}
\een
with a constant $m_n$ to be chosen so that $\zeta_n \in {\rm Diff}_+^{1,{\rm ps}}$. 
Now we form the corresponding coherent state $\omega'_n$ associated with $\zeta_n$ and we consider the relative entropy $S_{(-1/n,\infty)}(\omega | \omega'_n)$ relative to the algebra $\gM_{\rm Vir}(-1/n,\infty) \supset \gM_{\rm Vir}(0,\infty)$. By the monotonicity 
property of the relative entropy, we have $S_{(-1/n,\infty)}(\omega | \omega'_n) \ge S_{(0,\infty)}(\omega | \omega'_n)$. If we mollify $\zeta_n$ at $u=0$,  $S_{(-1/n,\infty)}(\omega | \omega'_n)$ can be computed using formula \eqref{eq:S0}, since $\zeta_n'(-1/n)=1, \zeta_n(-1/n)=0$, while $S_{(0,\infty)}(\omega | \omega'_n)$ can be computed using \eqref{eq:Snu}. Taking the mollifier away, 
we obtain the inequality
\ben
\nu(r) + \frac{c}{24}  \int_{0^+}^\infty u \left( \frac{\d}{\d u} \log \zeta_n'(u) \right)^2 \d u \le \frac{c}{24}  \int_{-1/n^+}^\infty (u+1/n) \left( \frac{\d}{\d u} \log \zeta_n'(u) \right)^2 \d u .  
\een
Evaluating the integrals is easy and leads to the bound
\ben
\nu(r) \le \frac{c(\log r)^2}{24} \left( \frac{n^{2/n-1}}{2} + \frac{1-n}{n(\log n)^2} \right) . 
\een
Taking $n \to \infty$, we obtain the second inequality: $\nu(r) \le 0$, therefore $\nu(r)=0$ in \eqref{eq:Snu}. Our final result \eqref{eq:Snu}can be again be expressed in terms fo the function $\varphi(s) = \log \rho(e^s)$ by a simple change of variables. Then we obtain the following theorem:

\begin{thm}
\label{thm:2}
The conclusions of thm. \ref{thm:1} remain true for 
$f \in C^\infty_{0,\RR}(\RR)$, with $f(0)=0$ but possibly, $f'(0) \neq 0$. 
\end{thm}

As an application, we consider the ``solitonic'' states defined and studied 
recently in \cite{delVeccio}, see \cite{frohlich} for a discussion of such states in specific models. In our language, these are coherent states
defined by a diffeomorphism $\rho$ which are smooth everywhere except at the point at infinity $u=\infty$, 
and which act trivially in a neighborhood of $\RR_-$. The derivative $r$ at the point at infinity is an invariant of the solitonic sector. 
Performing a $GNS$-representation of  $\gM_{\rm Vir}$ for the solitonic states, one obtains a representation which is unitarily inequivalent
to the vacuum representation for $\log r \neq 0$, and in which $L_0$ is represented by a non-negative 
self-adjoint operator (so these states have in particular finite total energy).  Since 
it is possible to exchange the points $u=0$ and $u=\infty$ by the M\" obius transformation $u \to 1/u$, we can apply thm. \ref{thm:2} 
to the solitonic states and we conclude:

\begin{corollary}
\label{cor:1}
Let $\omega$ be the vacuum state on $\gM_{\rm Vir}$ (or a finite index extension) 
and $\omega'$ be a solitonic state \cite{delVeccio}. 
Then $S(\omega | \omega') = c I_{\rm Schwarz}(\varphi)$, with the same notations as in thm. \ref{thm:1}.
\end{corollary}

\medskip
\noindent
{\bf Acknowledgements:} Part of this work was carried out while visiting IHES, Paris. It is a pleasure to thank IHES for hospitality and financial assistance. I am grateful to A. Connes  and R. Longo for discussions, and to the Max-Planck Society for supporting the collaboration between MPI-MiS and Leipzig U., grant Proj. Bez. M.FE.A.MATN0003. 

%bibliography

\end{document}